\documentclass[lettersize,journal]{IEEEtran}
\usepackage[utf8]{inputenc}
\usepackage{textcomp}
\usepackage{graphicx}
\usepackage{amsmath,amsfonts}

\usepackage{amsthm}
\usepackage{siunitx}
\usepackage{booktabs}
\usepackage{array}
\usepackage{stfloats}
\usepackage{url}
\usepackage{cite}
\usepackage[hidelinks,breaklinks]{hyperref}

\usepackage{todonotes}
\usepackage[dvipsnames]{xcolor}

\theoremstyle{plain}
\newtheorem{theorem}{Theorem}

\newtheorem{lemma}[theorem]{Lemma}
\theoremstyle{definition}
\newtheorem{assumption}{Assumption}
\newtheorem{remark}{Remark}
\newtheorem{problem}{Problem}

\newcommand{\R}{\mathbb{R}}
\newcommand{\N}{\mathbb{N}}
\newcommand{\Acl}{\mathbf{A}_{\mathrm{cl}}}
\newcommand{\Wset}{\mathcal{W}}
\newcommand{\Xset}{\mathcal{X}}
\newcommand{\Uset}{\mathcal{U}}
\newcommand{\Cset}{\mathcal{C}}
\newcommand{\Oinf}{\mathcal{O}_\infty}
\newcommand{\Rset}{\mathcal{R}}
\newcommand{\Pset}{\mathcal{P}}
\newcommand{\oplusm}{\oplus}

\begin{document}

\title{Certifiable Explicit Model Predictive Control for Spacecraft Rendezvous under Bounded Disturbances}

\author{Diogo Silva and Daniel Silvestre%
\thanks{This work was supported by FCT -- Funda\c{c}\~{a}o para a Ci\^{e}ncia e a Tecnologia, I.P., under project references UID/PRR/00066/2025 with DOI identifier \url{https://doi.org/10.54499/UID/PRR/00066/2025} and UID/00066/2025 with DOI identifier \url{https://doi.org/10.54499/UID/00066/2025}.}%
\thanks{D. Silva is with the School of Science and Technology from the NOVA University of Lisbon, 2829-516 Caparica, Portugal, with the Center for Technology and Systems (CTS) part of LASI, and also with the Institute for Systems and Robotics, Instituto Superior T\'{e}cnico, University of Lisbon. Email: \href{mailto:dgn.silva@campus.fct.unl.pt}{dgn.silva@campus.fct.unl.pt} \emph{(Corresponding author)}}%
\thanks{D. Silvestre is with the School of Science and Technology from the NOVA University of Lisbon, 2829-516 Caparica, Portugal, with the Center for Technology and Systems (CTS) part of LASI, and also with the Institute for Systems and Robotics, Instituto Superior T\'{e}cnico, University
of Lisbon. Email: \href{mailto:dsilvestre@fct.unl.pt}{dsilvestre@fct.unl.pt}}}

\markboth{IEEE Transactions on Control Systems Technology}%
{Silva \MakeLowercase{\textit{et al.}}: Certifiable Explicit MPC for Spacecraft Rendezvous}

\maketitle

\begin{abstract}
	Two properties of optimization-based controllers such as model predictive control (MPC) limit their use in space flight. The online computing time varies and can exceed the sampling period, and the closed-loop behavior carries no formal guarantee. Explicit MPC, where a parametric solution of the optimization problem is computed, replaces the on-board optimization with a lookup table of piecewise-affine functions with a fixed execution time. However, it has been regarded as only applicable to small horizons, since the number of regions bounds the memory the table occupies and admits only an exponential bound in the horizon. This paper builds on recent developments in parametric solutions of quadratic programs and nonconvex reachability analysis to provide an entire pipeline for spacecraft rendezvous controllers under the Clohessy-Wiltshire dynamics. For this application, efficient data structures can be constructed to encode the control law, both from a computing time and a memory standpoint. In addition, the controller can be verified offline with a guaranteed closed-loop stability certificate by resorting to reachability analysis with hybrid zonotopes. At a 51-step horizon the rendezvous law occupies \SI{10.5}{\mega\byte} and is evaluated in under \SI{300}{\nano\second}. The occupied memory would fit for instance the CubeSat flight computer ARM9-class AT91RM9200 meaning that the horizon limitation does not hold when the partition is built with the parametric solvers and data structures shown in this paper.
\end{abstract}


% =====================================================================
\section{Introduction}
	\label{sec:intro}

	\IEEEPARstart{A}{utonomous} rendezvous is the capability required for on-orbit servicing, active debris removal and in-space assembly, with a focus from the community on optimization-based controllers for their ability to encode constraints and result in optimized control actions \cite{SILVESTRE2026404}. The main objective is to enable a spacecraft to run a guidance law on its own processor without a ground loop inside the decision \cite{quartullo2026robust}. Two requirements are placed on such a law and they pull in different directions. It must generate constrained, propellant-efficient motion in real time on a space-qualified processor, and its closed-loop behavior must be established before launch over the entire navigation dispersion at handover and not over a collection of simulated cases. Therefore, one of the requirements aims for simple solutions for validation whereas the other requires complex algorithms to achieve optimized actions.

	Generating the guidance trajectory on board is attractive but comes with the price of added complexity. Direct and indirect optimal-control methods produce high-fidelity trajectories yet offer either no guaranteed convergence or a computational demand that a flight processor cannot accommodate, which has historically prevented them from being used on space missions \cite{bernardini2024state}. With recent advancements in convex optimization, this application has been revisited since a convex program admits a solution in polynomial time with a convergence bound known before the solver runs. Part of the nonconvexity in powered flight turns out to be exactly removable, because the thrust lower bound and pointing constraints admit a relaxation whose solution provably solves the original problem \cite{blackmore2010minimum}. A customized interior-point implementation of this relaxation was flight-tested on a terrestrial rocket in 2013 \cite{dueri2017customized}.

	The aforementioned strategy is not a general option as nonlinear dynamics, attitude-translation coupling, keep-out zones and plume impingement have no known lossless convexifications, which led the field to consider successive convexification, which replaces the problem by a sequence of convex subproblems iterated until convergence or until a computing time budget is exhausted \cite{mao2016successive}. These approaches have been surveyed for space vehicle control identifying lossless convexification, sequential convex programming and model predictive control (MPC) as key for future missions \cite{malyuta2021advances,wang2024survey}, and the rendezvous literature in particular has taken up sequential convex programming for six-degree-of-freedom close proximity \cite{zhou2019receding} and for docking under state-triggered constraints \cite{zhang2022trajectory}. For some nonconvexities like deadband, there are alternatives to reduce the computational time by exploiting the structure of the problem with approximate solutions to the optimal problem being found close to real time \cite{10542354}.

	Lossless convexification bounds the solve before it runs, whereas a successive scheme guarantees only that the sequence of subproblems converges. As such, the on-board cost is an iteration count capped in practice rather than bounded a priori, and the flight software must carry a solver together with its numerical failure modes. MPC in its implicit form inherits the same solver and the same failure modes. It is by now standard for proximity operations \cite{dicairano2012mpc,weiss2015mpc,hartley2012mpc} and includes formulations that recover the classical V-bar and R-bar hopping strategies of operational rendezvous inside the optimization itself \cite{silvaeurognc}.

	Explicit MPC instead solves the problem offline once as a parametric optimization over the whole parameter space, so that what runs online is a piecewise-affine lookup table with a fixed and data-independent execution time \cite{bemporad2002explicit,alessio2009survey,bemporad2021empc}. The work in \cite{dicairano2012mpc} applies such a law to the short-range phase, and \cite{leomanni2014explicit} to low-thrust proximity operations using a Laguerre parametrization. The objection raised against it is the offline size of that lookup, since the number of regions bounds the memory the table occupies and admits only an exponential bound in the horizon \cite{bemporad2021empc,alessio2009survey}. That objection has driven the rendezvous literature toward reducing the decision-variable count by other means, through parametric shrinking-horizon schemes \cite{castroviejo2024parametric,farooqi2020shrinking} or robust variable-horizon formulations for rotating targets \cite{quartullo2026robust}, all of which retain an online solver. These parametrizations have been compared on the region count alone, and not on the quantities that decide a flight implementation. As such, the first problem addressed in this paper is that comparison for the explicit rendezvous problem, on the memory of the evaluated law, the offline construction time, the closed-loop cost at matched propellant, and constraint compliance against a nonlinear plant.

	Reducing the number of free inputs shrinks the explicit partition, a mechanism identified in \cite{tondel2002complexity} and surveyed in \cite{rossiter2023review,valencia2011alternative}. The Laguerre construction is due to \cite{wang2004laguerre}, where discrete Laguerre functions form an orthonormal basis with an exponentially decaying envelope set by a single pole parameter, so a short sequence of them can represent an input profile that would otherwise need many free samples. Its appeal is argued in terms of the feasibility-performance compromise \cite{wang2008exploiting}, since at a fixed number of degrees of freedom a Laguerre basis spreads the input over a longer effective profile. Thereby, it admits a larger controlled admissible set than the conventional truncation, at the price of worse conditioning, which \cite{rossiter2023review} identifies as the reason the same basis inflates the number of critical regions. Move-blocking holds the input constant over blocks of samples and loses the a~priori stability guarantee \cite{makarow2024moveblock}. Complementary routes act on the partition after the fact, by clipping, merging regions, or restricting the parameter set to what the closed loop can reach \cite{kvasnica2019complexity}, and recent work replaces the polyhedral partition altogether with constrained-zonotope representations \cite{mihai2025empczono}, while efficient multi-parametric solvers \cite{arnstrom2024daqp} have considerably reduced the offline cost of building any of them.

	On the verification side, embedding the Karush-Kuhn-Tucker (KKT) conditions of the multi-parametric program in a hybrid zonotope allows computing exactly the closed-loop reachable set of a linear MPC, with memory growing linearly in the horizon \cite{bird2022hz}. The validation procedure proposed here is based on verifying that the evolution of the forward propagated set will become equivalent to the forward propagation using the unconstrained linear quadratic regulator (LQR) on a neighborhood of the origin \cite{scokaert1998clqr}. This is related to the maximal-admissible-set recursion \cite{gilbert1991maximal}, in the form that tightens each step of the recursion by the worst case over a bounded disturbance set \cite{kolmanovsky1998disturbance}. Regions of this kind act as terminal constraints in the classical stability arguments for MPC \cite{mayne2000constrained}, although the guarantee they carry is asymptotic, whereas the certificate established here holds after a finite number of time steps.

	The verification method presented here avoids validating closed-loop behavior of an optimization-based guidance law through the conventional approach of performing a Monte Carlo simulation \cite{hartley2012mpc,gavilan2012chance}. Sampling explores the reachable set from the \emph{inside}, so a campaign that finds no violation establishes no certificate of stability. Set-propagation methods that over-approximate fail in the opposite way, accumulating conservatism at every step until containment becomes unprovable for a reason that has nothing to do with the controller. The hybrid zonotope encoding of the multi-parametric program \cite{bird2022hz} avoids both, since the propagated set is exact.

	\subsection{Contributions and Organization}
    In this paper, the main contributions are the following:
		\begin{enumerate}
			\item A pipeline that answers the memory objection to explicit MPC for spacecraft rendezvous. Among the input parametrizations that reduce explicit complexity, namely the control horizon, Laguerre bases and move-blocking, a short control horizon lowers the closed-loop cost by up to \SI{22}{\percent} against the standard design and builds three times faster. The resulting lookup is \SI{1.85}{\mega\byte} at $N=16$ and \SI{10.5}{\mega\byte} at $N=51$, within the memory of a CubeSat flight computer \cite{araujo2025cots}, at under \SI{300}{\nano\second} per evaluation.
			\item A closed-form expression, valid for any linear input parametrization, for a region on which the applied input of the MPC law coincides with a single constant gain. This expression can be used for verification and validation with a negligible compute time since it is read from the condensed explicit solution of the quadratic program.
			\item A method to formally verify the closed-loop behavior, asserting constraint satisfaction and stability from hybrid zonotope propagation and a stopping criterion given by the maximal robust invariant subset of the constant-gain region.
		\end{enumerate}

		The rest of the paper is organized as follows. Section~\ref{sec:problem} states the relative dynamics, the constraints and the disturbance model of the rendezvous task, and separates the two problems the paper addresses, namely running the controller on board and verifying it before launch. Section~\ref{sec:empc} states the parametrized MPC law and its explicit solution, on which the first problem is posed. The second problem occupies the three sections that follow, with Section~\ref{sec:cgr} deriving the constant-gain region in closed form, Section~\ref{sec:cert} building the robust invariant set inside it and proving the finite-step reduction, and Section~\ref{sec:hz} giving the reachability computation and the containment test that decides it. Section~\ref{sec:results} shows the simulation results for a low-Earth-orbit rendezvous.

	\emph{Notation:} Lowercase bold denotes vectors ($\mathbf{x}$), uppercase bold denotes matrices ($\mathbf{A}$), and calligraphic letters denote sets ($\Cset_T$), while scalars and indices are plain ($N$, $k$, $T_s$). A bold $\mathbf{0}$ or $\mathbf{I}$ denotes the zero or identity matrix of the size implied by context. $\mathrm{diag}(\cdot)$ builds a block-diagonal matrix from its arguments, each of which may be a scalar or a matrix. A superscript star marks an optimal value, so $\boldsymbol{\eta}^\star$ is the minimizer of the optimization problem at hand and $\mathbf{u}^\star$ the input it produces, while the superscript $\mathrm{unc}$ marks the corresponding unconstrained value. The symbol $\oplusm$ denotes the Minkowski sum. An absolute value or an inequality between vectors is applied componentwise.

% =====================================================================
\section{Problem Formulation}
	\label{sec:problem}

	The two requirements that opened Section~\ref{sec:intro} are made precise here, first as a relative-motion model with its constraints and disturbance set, then as the two problems the rest of the paper tackles.

	\subsection{Relative dynamics and disturbance model}
		\label{sec:relmodel}
		Consider a chaser spacecraft performing a rendezvous with a passive target on a circular orbit of semi-major axis $a$ and mean motion $n = \sqrt{\mu/a^3}$. Relative motion is expressed in the target's local-vertical/local-horizontal (LVLH) frame and modeled by the Clohessy-Wiltshire (CW) equations \cite{clohessy1960terminal}, which give the discrete-time linear plant
		\begin{equation}
			\label{eq:plant}
			\mathbf{x}_{k+1} = \mathbf{A} \mathbf{x}_k + \mathbf{B} \mathbf{u}_k + \mathbf{w}_k , \qquad \mathbf{w}_k \in \Wset ,
		\end{equation}
		with state $\mathbf{x} \in \R^{n_x}$, input $\mathbf{u} \in \R^{n_u}$ and an additive disturbance confined to a set $\Wset$, subject to polytopic state and input constraints
		\begin{equation}
			\label{eq:constraints}
			\mathbf{x}_k \in \Xset , \qquad \mathbf{u}_k \in \Uset ,
		\end{equation}
		together with a terminal set $\Xset_N \subseteq \Xset$ and a dispersion $\Xset_0$ of admissible initial conditions. 

		Conventions for this frame differ between authors, so the one used here is made explicit. Writing $\mathbf{r}$ and $\mathbf{v}$ for the target's inertial position and velocity and $\mathbf{h}_{\mathrm{orb}}$ $= \mathbf{r} \times \mathbf{v}$ for its angular momentum, the frame axes are
		\begin{equation}
			\label{eq:lvlh}
			\mathbf{e}_z = -\frac{\mathbf{r}}{\lVert \mathbf{r} \rVert} , \qquad
			\mathbf{e}_y = -\frac{\mathbf{h}_{\mathrm{orb}}}{\lVert \mathbf{h}_{\mathrm{orb}} \rVert} , \qquad
			\mathbf{e}_x = \mathbf{e}_y \times \mathbf{e}_z ,
		\end{equation}
		so that $z$ points along the nadir, $y$ is the negative orbit normal, and $x$ completes the right-handed triad and coincides with the velocity direction for a circular orbit. The state is ordered as along-track, cross-track and nadir position followed by the corresponding velocities.

		With impulsive actuation and a sampling period $T_s$, this instantiates \eqref{eq:plant} with $\mathbf{x} = [x,y,z,v_x,v_y,v_z]^\top \in \R^{n_x}$ collecting relative position and velocity, $\mathbf{u} \in \R^{n_u}$ the commanded velocity increment, and $\mathbf{A} = \mathbf{A}_{\mathrm{CW}}$, $\mathbf{B} = \mathbf{B}_{\mathrm{CW}}$ the exact discretization of the CW state-transition matrix over $T_s$. Writing $\sigma := n T_s$, $s := \sin\sigma$, $c := \cos\sigma$,
		\begin{equation}
			\label{eq:cwmat}
			\mathbf{A}_{\mathrm{CW}} = \begin{bmatrix}
				1 & 0 & 6(\sigma-s) & \tfrac{4s}{n}-3T_s & 0 & \tfrac{2(1-c)}{n} \\
				0 & c & 0 & 0 & \tfrac{s}{n} & 0 \\
				0 & 0 & 4-3c & -\tfrac{2(1-c)}{n} & 0 & \tfrac{s}{n} \\
				0 & 0 & 6n(1-c) & 4c-3 & 0 & 2s \\
				0 & -ns & 0 & 0 & c & 0 \\
				0 & 0 & 3ns & -2s & 0 & c
			\end{bmatrix} ,
		\end{equation}
		rows and columns ordered as $\mathbf{x}$, where the zero pattern is exactly the in-plane/out-of-plane decoupling (cross-track motion $(y,v_y)$ never mixes with the rest). Because actuation is assumed to be an impulsive velocity increment, $\mathbf{B}_{\mathrm{CW}}$ is the three columns of $\mathbf{A}_{\mathrm{CW}}$ acting on $(v_x,v_y,v_z)$. The CW subscript is dropped from here onward, since no other plant model is used.

		\eqref{eq:cwmat} is a first-order expansion about a circular reference orbit, so over one sampling period the state it predicts differs from the true motion. That difference collects the linearization error, which grows with the separation from the target, the eccentricity of the target orbit, whose correct linearization is time-varying, and every perturbation the model omits, of which differential $J_2$ is the largest at this altitude. Together with any actuation execution error, these are collected in $\mathbf{w}_k$, on which Theorem~\ref{thm:cert} requires only a bound over a single sampling period.

		\begin{remark}[Scope of the linear model]
			\label{rem:scope}
			The difference grows with the separation, from \SI{0.31}{\meter} over one sample at the \SI{2}{\kilo\meter} start of the transfer to \SI{115}{\meter} at \SI{100}{\kilo\meter} in Fig.~\ref{fig:cwbudget}, and a shorter sampling period does not recover it. The regime treated here is therefore the close-range phase with a near-circular target, whereas a more distant or eccentric rendezvous calls for a different model of relative motion \cite{sullivan2017survey}. Lemma~\ref{thm:cgr} and Theorem~\ref{thm:cert} are stated for the linear time-invariant plant of \eqref{eq:plant}, so any other linear time-invariant model may be substituted for \eqref{eq:cwmat} without changing them.
		\end{remark}

		\begin{assumption}
			\label{as:W} $\Wset \subset \R^{n_x}$ is known, compact and convex, with $\mathbf{0} \in \Wset$.
		\end{assumption}

		Since the guarantees sought here are deterministic, no distribution needs to be assumed for the disturbances, with the nominal case being recovered by $\Wset = \{\mathbf{0}\}$.

	\subsection{Constraints and the rendezvous task}
		The chaser must remain inside an approach corridor and respect actuator limits, which fixes the sets of \eqref{eq:constraints} as the boxes
		\begin{equation}
			\label{eq:constraintsCW}
			\Xset := \{\mathbf{x} : \underline{\mathbf{x}} \leq \mathbf{x} \leq \overline{\mathbf{x}}\}, \qquad
			\Uset := \{\mathbf{u} : \lvert \mathbf{u} \rvert \leq \overline{\mathbf{u}}\} ,
		\end{equation}
		and must be driven from a dispersed initial condition $\mathbf{x}_0 \in \Xset_0$ to a terminal box $\Xset_N \subset \Xset$ around the docking point. The set $\Xset_0$ is the navigation dispersion at handover, over which stability must be guaranteed for \emph{every} initial condition and \emph{every} disturbance sequence in $\Wset$, and not only for the nominal trajectory.

		The disturbance acts after the input, so a state admissible to the program at one step is displaced by at most one $\mathbf{w} \in \Wset$ before the next decision is taken. The corridor has to absorb that displacement, since a state that only just satisfies the box is carried across it by the disturbance alone.

		\begin{assumption}
			\label{as:margin} The program of \eqref{eq:ocp} is posed on the shrunk corridor $\Xset \ominus \Wset$, which is non-empty.
		\end{assumption}

		Both the plant and the sets above are written in coordinates centered on the docking point, $\mathbf{x} \mapsto \mathbf{x} - \mathbf{x}_{\mathrm{tgt}}$ with $\mathbf{x}_{\mathrm{tgt}}$ the target state of Table~\ref{tab:scenario}. A pure along-track offset at zero relative velocity is an equilibrium of \eqref{eq:cwmat}, so the translation leaves $\mathbf{A}$ and $\mathbf{B}$ unchanged and moves only the constraint boxes of \eqref{eq:constraintsCW} and the dispersion $\Xset_0$. The regulation objective is then stated about the origin of the shifted frame, which places $\mathbf{0}$ in the interior of $\Cset_T$ as property~\ref{it:rpi-fin} of Theorem~\ref{thm:rpi} requires, and every set constructed in Sections~\ref{sec:cgr} to~\ref{sec:hz} is expressed in that frame.

	\subsection{Onboard Controller and Offline Verification Problems}
		\label{sec:twoproblems}
		The control problem is to drive the chaser from anywhere in $\Xset_0$ to $\Xset_N$ while respecting \eqref{eq:constraintsCW}, under any disturbance sequence admitted by Assumption~\ref{as:W}. The approach corridor and the per-axis impulse bound are hard limits that hold at every step, and propellant is what the transfer spends, so the law adopted here encodes the constraints and the cost together in its own formulation.

		A stored control law is defined only over the region it was built on, so the transfer must stay inside that region for the stored law to apply at every step.

		\begin{assumption}
			\label{as:theta} The region over which the explicit solution is built covers every state the closed loop can visit, under every admissible disturbance.
		\end{assumption}

		Two things about the resulting law then remain open, namely whether it can be evaluated on a space-qualified processor and whether its behavior over the transfer can be established before launch.

		\begin{problem}[On-board implementation]
			\label{prob:impl}
			Choose the input parametrization $\mathbf{T}$ of \eqref{eq:param} so that the explicit solution of the resulting program is stored and evaluated within the budget of a flight processor,
			\begin{equation}
				\label{eq:prob1}
				\mathrm{mem}(\mathbf{T}) \ \leq \ M_{\mathrm{avail}} , \qquad
				t_{\mathrm{eval}}(\mathbf{T}) \ \leq \ T_s ,
			\end{equation}
			where $\mathrm{mem}$ is the memory the stored law occupies, $M_{\mathrm{avail}}$ the memory the processor carries, and $t_{\mathrm{eval}}$ the time one evaluation takes, which must fit inside the sampling period, while the closed loop still drives every $\mathbf{x}_0 \in \Xset_0$ into $\Xset_N$ under \eqref{eq:constraintsCW}.
		\end{problem}

		\begin{problem}[Offline verification]
			\label{prob:verif}
			For the controller from Problem \ref{prob:impl}, establish that every $\mathbf{x}_0 \in \Xset_0$ and every disturbance sequence $\{\mathbf{w}_k\}_{k \geq 0} \subset \Wset$ produce a trajectory that respects the constraints at every step and converges to a bounded set $\Xset_\infty$ around the docking point, whose size the disturbance fixes and which is $\{\mathbf{0}\}$ when the model is exact,
			\begin{equation}
				\label{eq:probstatement}
				\begin{aligned}
					\mathbf{x}_k &\in \Xset , \quad \mathbf{u}_k \in \Uset , & &\forall k \geq 0 , \\
					\mathbf{x}_k &\to \Xset_\infty , & &k \to \infty .
				\end{aligned}
			\end{equation}
			The guarantee is to be established by computation on the exact closed-loop reachable sets rather than by sampling.
		\end{problem}

		The parametrization enters both problems, since it fixes the on-board footprint of Problem~\ref{prob:impl} and also the size of the reachability problem from which the guarantee of Problem~\ref{prob:verif} is computed.

% =====================================================================
\section{Explicit MPC for Flight Implementation}
	\label{sec:empc}

	Addressing Problem~\ref{prob:impl} requires an explicit MPC with the input sequence restricted to a low-dimensional subspace, a restriction that keeps its explicit solution small enough to store. That restriction is the only design freedom exercised in Problem~\ref{prob:impl}, every other element of the controller being held fixed across the choices compared in Section~\ref{sec:results}.

	\subsection{Parametrized MPC}
		\label{sec:mpc} At each step the controller solves, over the horizon $N$,
		\begin{equation}
			\label{eq:ocp}
			\begin{aligned}
				\min_{\mathbf{U}} \ & \sum_{i=0}^{N-1}\left( \mathbf{x}_i^\top \mathbf{Q} \mathbf{x}_i + \mathbf{u}_i^\top \mathbf{R} \mathbf{u}_i \right) + \mathbf{x}_N^\top \mathbf{Q}_N \mathbf{x}_N \\
				\text{s.t.} \ \ & \eqref{eq:plant}\ \text{with}\ \mathbf{w} \equiv 0,\ \eqref{eq:constraints},\ \mathbf{x}_N \in \Xset_N ,
			\end{aligned}
		\end{equation}
		with $\mathbf{Q} \succeq 0$ and $\mathbf{R} \succ 0$ chosen by the designer. The terminal weight $\mathbf{Q}_N$ is not free, being obtained by solving the discrete algebraic Riccati equation defined by $\mathbf{A}$, $\mathbf{B}$, $\mathbf{Q}$ and $\mathbf{R}$. Stacking the horizon into $\mathbf{X} = [\mathbf{x}_1^\top \cdots \mathbf{x}_N^\top]^\top$ gives $\mathbf{X} = \hat{\mathbf{A}} \mathbf{x}_0 + \hat{\mathbf{B}} \mathbf{U}$ for the lifted prediction matrices
		\begin{equation}
			\begin{gathered}
				\hat{\mathbf{A}} = \begin{bmatrix} \mathbf{A} \\ \mathbf{A}^2 \\ \vdots \\ \mathbf{A}^N \end{bmatrix} \in \R^{Nn_x \times n_x} , \\[4pt]
				\hat{\mathbf{B}} = \begin{bmatrix} \mathbf{B} & \mathbf{0} & \cdots & \mathbf{0} \\ \mathbf{A}\mathbf{B} & \mathbf{B} & \ddots & \vdots \\ \vdots & \ddots & \ddots & \mathbf{0} \\ \mathbf{A}^{N-1}\mathbf{B} & \cdots & \mathbf{A}\mathbf{B} & \mathbf{B} \end{bmatrix} \in \R^{Nn_x \times Nn_u} ,
			\end{gathered}
		\end{equation}
		and $\hat{\mathbf{Q}} = \mathrm{diag}(\mathbf{Q},\dots,\mathbf{Q},\mathbf{Q}_N) \in \R^{Nn_x \times Nn_x}$, $\hat{\mathbf{R}} = \mathrm{diag}(\mathbf{R},\dots,\mathbf{R}) \in \R^{Nn_u \times Nn_u}$ collect the stage weights over the horizon. Substituting $\mathbf{X}$ into \eqref{eq:ocp} and eliminating the states gives the condensed form in $\mathbf{U} = [\mathbf{u}_0^\top \cdots \mathbf{u}_{N-1}^\top]^\top$,
		\begin{equation}
			\label{eq:condensed}
			\begin{aligned}
				\min_{\mathbf{U}} \ & \tfrac{1}{2} \mathbf{U}^\top \mathbf{H} \mathbf{U} + \mathbf{x}^\top \mathbf{F}^\top \mathbf{U} \\
				\text{s.t.} \ \ & \mathbf{G} \mathbf{U} \leq \mathbf{g} + \mathbf{S} \mathbf{x} ,
			\end{aligned}
		\end{equation}
		where $\mathbf{H} = \hat{\mathbf{B}}^\top \hat{\mathbf{Q}} \hat{\mathbf{B}} + \hat{\mathbf{R}} \succ 0$ and $\mathbf{F} = \hat{\mathbf{B}}^\top \hat{\mathbf{Q}} \hat{\mathbf{A}}$ follow from the prediction matrices, and $\mathbf{G} \in \R^{n_c \times Nn_u}$, $\mathbf{g} \in \R^{n_c}$, $\mathbf{S} \in \R^{n_c \times n_x}$ stack the state, terminal and input constraints.

		Input parametrization restricts the decision variable to a subspace of $n_\eta$ free variables,
		\begin{equation}
			\label{eq:param}
			\mathbf{U} = \mathbf{T} \boldsymbol{\eta} , \qquad \mathbf{T} \in \R^{Nn_u \times n_\eta}, \qquad n_\eta \leq N n_u ,
		\end{equation}
		which turns \eqref{eq:condensed} into a smaller quadratic program with
		\begin{equation}
			\label{eq:projected}
			\mathbf{H}_T = \mathbf{T}^\top \mathbf{H} \mathbf{T}, \qquad \mathbf{F}_T = \mathbf{T}^\top \mathbf{F}, \qquad \mathbf{G}_T = \mathbf{G} \mathbf{T} .
		\end{equation}
		Four choices are considered, all of them standard \cite{rossiter2023review,valencia2011alternative},
		\begin{equation}
			\label{eq:bases}
			\begin{aligned}
				\text{standard:} \quad && \mathbf{T} &= \mathbf{I} , & n_\eta &= N n_u , \\
				\text{control horizon:} \quad && \mathbf{T}_{N_u} &= \begin{bmatrix} \mathbf{I}_{N_u n_u} \\ \mathbf{0} \end{bmatrix} , & n_\eta &= N_u n_u , \\
				\text{Laguerre:} \quad && \mathbf{T}_{\mathrm{Lag}} &= \mathbf{L}^\top \otimes \mathbf{I}_{n_u} , & n_\eta &= n_L n_u , \\
				\text{move-blocking:} \quad && \mathbf{T}_{\mathrm{MB}} &= \mathbf{M} \otimes \mathbf{I}_{n_u} , & n_\eta &= n_{\mathrm{MB}} n_u ,
			\end{aligned}
		\end{equation}
		where $\mathbf{T}_{N_u}$ leaves the first $N_u < N$ moves free and fixes the remainder to zero, $\mathbf{L} \in \R^{n_L \times N}$ is built from $n_L$ discrete Laguerre functions of pole $a_L \in [0,1)$ \cite{wang2004laguerre,wang2008exploiting}, trading a longer effective input profile for the same number of variables, and $\mathbf{M} \in \R^{N \times n_{\mathrm{MB}}}$ is a zero-one blocking matrix holding the input constant over $n_{\mathrm{MB}}$ blocks \cite{tondel2002complexity,makarow2024moveblock}.

		The four bases are related, since at zero pole the Laguerre basis degenerates to the control horizon,
		\begin{equation}
			\label{eq:degen}
			a_L = 0 \quad \Longrightarrow \quad
			\mathbf{L} = [\,\mathbf{I}_{n_L}\ \ \mathbf{0}\,] , \qquad
			\mathbf{T}_{\mathrm{Lag}} = \mathbf{T}_{N_u} \big|_{N_u = n_L} ,
		\end{equation}
		so the control horizon is the zero-pole member of the Laguerre family, and is named as such throughout.

	\subsection{Explicit solution}
		The optimization problem~\eqref{eq:condensed}-\eqref{eq:param} is a multi-parametric quadratic program in the parameter $\mathbf{x}$. Its solution is a continuous piecewise-affine map
		\begin{equation}
			\label{eq:pwa}
			\boldsymbol{\eta}^\star(\mathbf{x}) = \boldsymbol{\Gamma}_i \mathbf{x} + \boldsymbol{\gamma}_i \quad \text{for} \quad \mathbf{x} \in \Theta_i ,
			\qquad i = 1,\dots,N_r ,
		\end{equation}
		over a polyhedral partition $\{\Theta_i\}$ of $N_r$ regions whose union is the feasible subset of a declared parameter box $\Theta$, so that a state at which the law is evaluated lies in some $\Theta_i$ whenever the program is feasible there \cite{bemporad2002explicit,alessio2009survey}. On board, evaluation reduces to locating $\mathbf{x}$ in the partition and applying the corresponding affine law, which gives explicit MPC its data-independent execution time.

		A tree locates $\mathbf{x}$ in a number of tests that grows with the logarithm of the region count rather than with the count itself, so the partition is saved as a binary search tree \cite{arnstrom2024daqp}, whose internal nodes each hold one of the hyperplanes bounding the regions of \eqref{eq:pwa}. As such, evaluation descends the tree by testing $\mathbf{x}$ against one hyperplane per level until it reaches a leaf carrying the affine law to apply. The tree stores less than the full explicit solution, since only the first input block of \eqref{eq:pwa} is ever needed on board, the rest of the predicted sequence being discarded at the next step, and since leaves carrying identical laws are merged, so that it indexes distinct gains and not regions. Its depth therefore bounds the evaluation time $t_{\mathrm{eval}}$ of \eqref{eq:prob1} and its number of distinct leaves fixes the stored memory, both measured in Section~\ref{sec:memory}.

		In the notation of this section, Assumption~\ref{as:theta} reads $\Rset_k^\Wset \subseteq \Theta$ for all $k \geq 0$. It is verified by the same reachability computation used in Section~\ref{sec:hz}, so it requires no additional work, and sizing $\Theta$ from the reachable set instead of from the state box is also a complexity-reduction mechanism in its own right \cite{kvasnica2019complexity}.

		Problem~\ref{prob:impl} is therefore addressed by a parametrization $\mathbf{T}$, which fixes the size of the explicit solution, together with a search tree over the resulting partition. One level of that tree costs a single inner product in $\R^{n_x}$ and the leaf a single $n_u \times n_x$ product, neither of which grows with the horizon. This construction leaves open the behavior of the closed loop it produces, which is Problem~\ref{prob:verif} and occupies the three sections that follow.

% =====================================================================
\section{The Constant-Gain Region of a Parametrized MPC Law}
	\label{sec:cgr}

	Before the solution to Problem~\ref{prob:verif} is introduced, note that the applied input eventually equals a constant gain $\mathbf{K}_T$, which happens wherever the constraints of the underlying program are inactive at the optimum. This will happen in a neighborhood of the origin \cite{scokaert1998clqr}, where the constrained law coincides with the unconstrained minimizer of the parametrized program. For that reason, the remainder of this section computes that region for an arbitrary input parametrization, which then defines a stopping criterion for the verification procedure through the robust invariant set of Section~\ref{sec:cert}.

	Let $\mathbf{E}_0 = [\,\mathbf{I}_{n_u} \ \mathbf{0} \ \cdots \ \mathbf{0}\,] \in \R^{n_u \times N n_u}$ select the first input block, and define the unconstrained minimizer $\boldsymbol{\eta}^{\mathrm{unc}}$ of the parametrized program \eqref{eq:projected},
	\begin{equation}
		\label{eq:Pt}
		\begin{gathered}
			\mathbf{P}_T := \mathbf{H}_T^{-1} \mathbf{F}_T = (\mathbf{T}^\top \mathbf{H} \mathbf{T})^{-1} \mathbf{T}^\top \mathbf{F} \in \R^{n_\eta \times n_x} , \\[3pt]
			\boldsymbol{\eta}^{\mathrm{unc}}(\mathbf{x}) = -\mathbf{P}_T \mathbf{x} ,
		\end{gathered}
	\end{equation}
	together with the induced gain and closed-loop matrix
	\begin{equation}
		\label{eq:Kt}
		\mathbf{K}_T := \mathbf{E}_0 \mathbf{T} \mathbf{P}_T \in \R^{n_u \times n_x} , \qquad \Acl := \mathbf{A} - \mathbf{B} \mathbf{K}_T \in \R^{n_x \times n_x} .
	\end{equation}

	Each parametrization $\mathbf{T}$ induces its own gain $\mathbf{K}_T$, to which the applied input reduces. 

	\begin{lemma}[Constant-gain region, closed form]
		\label{thm:cgr} Let $\mathbf{H}_T \succ 0$ and define
		\begin{equation}
			\label{eq:C}
			\Cset_T := \left\{ \mathbf{x} \in \Xset \ : \ \left( -\mathbf{G}_T \mathbf{P}_T - \mathbf{S} \right) \mathbf{x} \leq \mathbf{g} \right\} .
		\end{equation}
		Then the parametrized program \eqref{eq:condensed}-\eqref{eq:param} attains its unconstrained minimizer on $\Cset_T$,
		\begin{equation}
			\label{eq:collapse}
			\begin{gathered}
				\mathbf{x} \in \Cset_T \quad \Longrightarrow \quad \boldsymbol{\eta}^\star(\mathbf{x}) = -\mathbf{P}_T \mathbf{x} , \\[3pt]
				\mathbf{u}^\star(\mathbf{x}) = \mathbf{E}_0 \mathbf{T} \boldsymbol{\eta}^\star(\mathbf{x}) = -\mathbf{K}_T \mathbf{x} ,
			\end{gathered}
		\end{equation}
		a single constant gain, and $\Cset_T$ is a polyhedron with at most $n_c + 2n_x$ facets.
	\end{lemma}

	\begin{proof}
		Since $\mathbf{H}_T \succ 0$ the program is strictly convex, so its minimizer is unique and coincides with the unconstrained one whenever the latter is feasible. Substituting $\boldsymbol{\eta}^{\mathrm{unc}} = -\mathbf{P}_T \mathbf{x}$ into the constraint of \eqref{eq:condensed} gives
		\begin{equation}
			\label{eq:Cderiv}
			\mathbf{G}_T \left( -\mathbf{P}_T \mathbf{x} \right) \ \leq \ \mathbf{g} + \mathbf{S}\mathbf{x}
			\qquad \Longleftrightarrow \qquad
			\left( -\mathbf{G}_T \mathbf{P}_T - \mathbf{S} \right) \mathbf{x} \ \leq \ \mathbf{g} ,
		\end{equation}
		whose second form is the description in \eqref{eq:C}, so feasibility of $\boldsymbol{\eta}^{\mathrm{unc}}$ is equivalent to $\mathbf{x} \in \Cset_T$, and \eqref{eq:collapse} follows with $\mathbf{K}_T$ from \eqref{eq:Kt}. The description \eqref{eq:C} intersects $n_c$ halfspaces with the box $\Xset$, giving the facet count.
	\end{proof}

	Lemma~\ref{thm:cgr} restates a classical fact, since $\Cset_T$ is the critical region associated with the empty active set, whose description follows the standard multi-parametric construction \cite{bemporad2002explicit}, and the coincidence of the constrained and unconstrained solutions follows from \cite{scokaert1998clqr}. Both the optimizer and the constraints are affine in $\mathbf{x}$, so $\Cset_T$ follows from the data of the condensed program by matrix multiplication alone. Equations~\eqref{eq:Pt}, \eqref{eq:Kt} and~\eqref{eq:C} apply unchanged to every parametrization, since only $(\mathbf{P}_T, \mathbf{G}_T, \mathbf{K}_T)$ change. As such, $\Cset_T$ is available for any $\mathbf{T}$ at negligible cost, and can be computed ahead of the reachability problem stated against it.

	\begin{remark}[Each basis reduces to its own gain]
		For $\mathbf{T} = \mathbf{I}$ the terminal weight $\mathbf{Q}_N$ is the Riccati solution, so $\mathbf{K}_T$ is the infinite-horizon LQR gain and $\Cset_T$ is the classical region where the constrained problem coincides with the unconstrained regulator \cite{scokaert1998clqr}. For $\mathbf{T} \neq \mathbf{I}$ this is no longer true, since each restricted basis reduces to \emph{its own} constant gain $\mathbf{K}_T$, which is not the LQR gain. This is why $\Cset_T$ is named a constant-gain region and not an LQR region.
	\end{remark}

% =====================================================================
\section{Robust Invariance and Finite-Step Reduction to the Constant Gain Region}
	\label{sec:cert}

	Section~\ref{sec:cgr} supplies the region $\Cset_T$ on which the applied control action equals $-\mathbf{K}_T \mathbf{x}$ \emph{at the current step}, although it is not necessarily invariant, so the successor $\Acl\mathbf{x}$ of a state inside it can fall outside. Invariance under a disturbance is stronger still, since every successor the disturbance can produce must return to the region, $\Acl \mathbf{x} + \mathbf{w} \in \Cset_T$ for every $\mathbf{w} \in \Wset$. Both are handled by the classical maximal-admissible-set construction \cite{gilbert1991maximal}, tightened at each step by the worst case over $\Wset$ \cite{kolmanovsky1998disturbance}, which admits a closed form here.

	For an autonomous linear system constrained to a polyhedron, the maximal admissible set is the largest set of initial states whose entire trajectory satisfies the constraint. It is itself a polyhedron, and it is finitely determined whenever the closed-loop matrix is Schur stable and the constraint set is bounded with the origin in its interior \cite{gilbert1991maximal}. Under an additive disturbance the same recursion applies with each constraint row tightened by the support function of the disturbance set, which yields the maximal robust positively invariant subset \cite{kolmanovsky1998disturbance}. Sets built this way are the classical route to stability guarantees in MPC, where they act as terminal constraints \cite{mayne2000constrained}, although the guarantee they carry there is asymptotic. Closed-form approximations avoid the recursion altogether at the price of a smaller set \cite{SILVESTRE2024126}, and the containment test of Section~\ref{sec:hz} applies to either, since it requires only a description of the set the reachable tube must lie inside.

	Write $h_\Wset(\mathbf{d}) := \max_{\mathbf{w} \in \Wset} \mathbf{d}^\top \mathbf{w}$ for the support function of $\Wset$, and let $\Cset_T = \{\mathbf{x} : \mathbf{D}\mathbf{x} \leq \mathbf{h}\}$, $\mathbf{D} \in \R^{n_f \times n_x}$.

	\begin{theorem}[Maximal robust positively invariant subset of $\Cset_T$]
		\label{thm:rpi} Let $\Pset_j$ collect the states admissible $j$ steps ahead against every disturbance sequence, and define $\Oinf^\Wset$ as their intersection,
		\begin{equation}
			\label{eq:Oinf}
			\begin{gathered}
				\Pset_j := \left\{ \mathbf{x} \ : \ \mathbf{D}\Acl^{j} \mathbf{x} \ \leq \ \mathbf{h} - \sum_{i=0}^{j-1} h_\Wset\!\left( (\mathbf{D}\Acl^{i})^\top \right) \right\} , \\[3pt]
				\Oinf^\Wset := \bigcap_{j \in \N_0} \Pset_j ,
			\end{gathered}
		\end{equation}
		where $h_\Wset$ acts row-wise, so that $\Pset_0 = \Cset_T$. Then
		\begin{enumerate}
			\item \label{it:rpi-sub} $\Oinf^\Wset \subseteq \Cset_T$;
			\item \label{it:rpi-inv} $\Acl \Oinf^\Wset \oplusm \Wset \subseteq \Oinf^\Wset$;
			\item \label{it:rpi-max} $\Oinf^\Wset = \bigcup \left\{ \mathcal{V} \subseteq \Cset_T \ : \ \Acl \mathcal{V} \oplusm \Wset \subseteq \mathcal{V} \right\}$, the union of every robustly invariant subset of $\Cset_T$;
			\item \label{it:rpi-fin} if $\Acl$ is Schur, $\Wset$ is bounded, $\Cset_T$ is bounded with $\mathbf{0} \in \operatorname{int} \Cset_T$, and the accumulated tightening leaves $\mathbf{0} \in \operatorname{int} \Oinf^\Wset$, the intersection is finitely determined,
			\begin{equation}
				\label{eq:Oinffin}
				\exists \, j^\star < \infty \ : \quad \Oinf^\Wset = \bigcap_{j=0}^{j^\star} \Pset_j .
			\end{equation}
		\end{enumerate}
		For a box disturbance set the support function is available in closed form,
		\begin{equation}
			\label{eq:boxsupp}
			\Wset = \left\{ \mathbf{w} : \lvert \mathbf{w} \rvert \leq \bar{\mathbf{w}} \right\}
			\quad \Longrightarrow \quad
			h_\Wset\!\left( (\mathbf{D}\Acl^{i})^\top \right) = \lvert \mathbf{D}\Acl^{i} \rvert \, \bar{\mathbf{w}} ,
		\end{equation}
		so \eqref{eq:Oinf} is evaluated by matrix multiplication alone.
	\end{theorem}

	\begin{proof}
		Iterating \eqref{eq:plant} under $\mathbf{u} = -\mathbf{K}_T \mathbf{x}$ gives
		\begin{equation}
			\label{eq:iterate}
			\mathbf{x}_j = \Acl^j \mathbf{x}_0 + \sum_{i=0}^{j-1}\Acl^{i} \mathbf{w}_{j-1-i} ,
		\end{equation}
		so imposing $\mathbf{x}_j \in \Cset_T$ for every admissible disturbance sequence reads, row-wise,
		\begin{equation}
			\label{eq:rowwise}
			\mathbf{D}\Acl^j \mathbf{x}_0 \ \leq \ \mathbf{h} - \max_{\mathbf{w} \in \Wset^{j}} \ \mathbf{D}
			\sum_{i=0}^{j-1} \Acl^{i} \mathbf{w}_{j-1-i} .
		\end{equation}
		The $\mathbf{w}_i$ are independent, so the maximum separates,
		\begin{equation}
			\label{eq:separate}
			\max_{\mathbf{w} \in \Wset^{j}} \ \mathbf{D} \sum_{i=0}^{j-1} \Acl^{i} \mathbf{w}_{j-1-i}
			\ = \ \sum_{i=0}^{j-1} h_\Wset\!\left( (\mathbf{D}\Acl^{i})^\top \right) ,
		\end{equation}
		so substituting \eqref{eq:separate} into \eqref{eq:rowwise} identifies the $j$-step admissible set with $\Pset_j$, and intersecting over $j \in \N_0$ gives \eqref{eq:Oinf}. Property~\ref{it:rpi-sub} is then immediate from $\Pset_0 = \Cset_T$. Re-indexing \eqref{eq:Oinf} in $j$ proves property~\ref{it:rpi-inv}, the standard argument of \cite{kolmanovsky1998disturbance}. For property~\ref{it:rpi-max}, any $\mathcal{V}$ in the union satisfies $\Acl^{j}\mathcal{V} \oplusm \left( \bigoplus_{i=0}^{j-1} \Acl^{i}\Wset \right) \subseteq \mathcal{V} \subseteq \Cset_T$ by induction, hence $\mathcal{V} \subseteq \Pset_j$ for every $j$ and so $\mathcal{V} \subseteq \Oinf^\Wset$. Conversely, $\Oinf^\Wset$ itself belongs to the union by properties~\ref{it:rpi-sub}-\ref{it:rpi-inv}, and the two inclusions give equality. For property~\ref{it:rpi-fin}, $\Acl$ Schur gives $\mathbf{D}\Acl^{j} \to \mathbf{0}$ geometrically while the tightening converges, so there is a finite $j^\star$ beyond which $\bigcap_{i=0}^{j^\star} \Pset_i \subseteq \Pset_j$ for all $j > j^\star$, which is \eqref{eq:Oinffin}, the finite-determination argument of \cite{gilbert1991maximal}. \eqref{eq:boxsupp} is the support function of a symmetric box. \qedhere
	\end{proof}

	Setting $\Wset = \{\mathbf{0}\}$ recovers the nominal construction of Gilbert and Tan \cite{gilbert1991maximal}, written $\Oinf$ throughout. The disturbance enters only through the tightening term, so the number of facets of $\Oinf^\Wset$ is typically unchanged, with only the right-hand side moving, until the tightening makes the set empty. An empty $\Oinf^\Wset$ means the disturbance is too large for the constant-gain region to retain any robustly invariant subset. As such, there is no step beyond which the applied input can be guaranteed to be a single fixed gain.

	\begin{remark}[Maximal within $\Cset_T$, not maximal]
		\label{rem:maximal} Larger robustly invariant sets of the closed loop generally exist. The piecewise-affine law \eqref{eq:pwa} induces a different closed-loop matrix $\mathbf{A} - \mathbf{B} \mathbf{E}_0\boldsymbol{\Gamma}_i$ on each critical region, and a set invariant under that switched system would extend well beyond $\Cset_T$, but the applied law on it is no longer a single gain, and obtaining it needs the enumerated partition, which Lemma~\ref{thm:cgr} was constructed to avoid. Since Theorem~\ref{thm:cert} states a sufficient condition, the larger set would lower $k^\star$ without changing what is guaranteed.
	\end{remark}

	It remains to define the tube along which those sets are propagated, so let $\Rset_0^\Wset := \Xset_0$ and let $\Rset_{k+1}^\Wset$ denote the one-step reachable set of the closed loop of \eqref{eq:plant} under the MPC law, including the disturbance,
	\begin{equation}
		\label{eq:tube}
		\Rset_{k+1}^\Wset = \Phi\!\left( \Rset_k^\Wset \right) \oplusm \Wset ,
	\end{equation}
	where $\Phi$ maps a set of states through one nominal closed-loop step. The nominal tube $\Rset_0 := \Xset_0$, $\Rset_{k+1} := \Phi(\Rset_k)$ is the case $\Wset = \{\mathbf{0}\}$. Section~\ref{sec:hz} makes $\Phi$ explicit.

	\begin{assumption}
		\label{as:feas} The program of \eqref{eq:condensed} is feasible $\forall \mathbf{x} \in \Rset_k^\Wset,\ \forall k \leq k^\star$.
	\end{assumption}

	Assumption~\ref{as:feas} is checkable by the same mixed-integer tests used for containment in Section~\ref{sec:hz}, and it is verified for the scenario of Section~\ref{sec:results}. Beyond $k^\star$ it is automatic, since $\Oinf^\Wset \subseteq \Cset_T$ and $\boldsymbol{\eta}^\star(\mathbf{x}) = -\mathbf{P}_T \mathbf{x}$ is feasible $\forall \mathbf{x} \in \Cset_T$ by construction.

	\begin{theorem}[Robust finite-step reduction to the constant gain]
		\label{thm:cert} Let Assumptions~\ref{as:W}, \ref{as:margin}, \ref{as:theta} and \ref{as:feas} hold and suppose there is a finite $k^\star$ with
		\begin{equation}
			\label{eq:certcond}
			\Rset_{k^\star}^\Wset \ \subseteq \ \Oinf^\Wset .
		\end{equation}
		Then, $\forall \mathbf{x}_0 \in \Xset_0$ and $\forall \{\mathbf{w}_k\}_{k \geq 0} \subset \Wset$,
		\begin{enumerate}
			\item \label{it:cert-gain} the optimizer returns the fixed gain, so the loop is linear,
			\begin{equation}
				\label{eq:certcollapse}
				\mathbf{u}_k = -\mathbf{K}_T \mathbf{x}_k , \qquad
				\mathbf{x}_{k+1} = \Acl \mathbf{x}_k + \mathbf{w}_k , \qquad \forall k \geq k^\star ;
			\end{equation}
			\item \label{it:cert-constr} $\mathbf{x}_k \in \Xset$ and $\mathbf{u}_k \in \Uset,\ \forall k \geq 0$;
			\item \label{it:cert-conv} if $\Acl$ is Schur, then $\mathbf{x}_k \to \Xset_\infty$ as $k \to \infty$, the set $\Xset_\infty$ of Problem~\ref{prob:verif} being the minimal robust positively invariant set of $(\Acl, \Wset)$.
		\end{enumerate}
	\end{theorem}

	\begin{proof}
		\eqref{eq:certcond} places $\mathbf{x}_{k^\star} \in \Oinf^\Wset$. For any $\mathbf{x}_k \in \Oinf^\Wset$, property~\ref{it:rpi-sub} of Theorem~\ref{thm:rpi} with Lemma~\ref{thm:cgr} fixes the applied law, and property~\ref{it:rpi-inv} returns the successor to the same set,
		\begin{equation}
			\label{eq:certstep}
			\mathbf{x}_k \in \Oinf^\Wset
			\ \Longrightarrow \
			\begin{gathered}
				\mathbf{u}_k = -\mathbf{K}_T \mathbf{x}_k \ \Longrightarrow \\[2pt]
				\mathbf{x}_{k+1} = \Acl \mathbf{x}_k + \mathbf{w}_k \in \Oinf^\Wset , \quad \forall \mathbf{w}_k \in \Wset ,
			\end{gathered}
		\end{equation}
		so induction from $k^\star$ gives property~\ref{it:cert-gain}. Property~\ref{it:cert-constr} follows, for $k \geq k^\star$, from $\mathbf{x}_k \in \Oinf^\Wset \subseteq \Cset_T \subseteq \Xset$ for the state, and from the input rows of $\mathbf{G}$ being among those defining $\Cset_T$ in \eqref{eq:C}, which give $\mathbf{u}_k \in \Uset$. For $k < k^\star$ it follows from Assumptions~\ref{as:margin} and~\ref{as:feas}, since a feasible program of \eqref{eq:condensed} places the predicted successor in $\Xset \ominus \Wset$, leaving the realized one in $\Xset$ for every $\mathbf{w}_k \in \Wset$. Property~\ref{it:cert-conv} is the standard convergence result for a Schur linear system driven by a bounded disturbance \cite{kolmanovsky1998disturbance}. \qedhere
	\end{proof}

	Two features of Theorem~\ref{thm:cert} are used repeatedly in what follows. First, the guarantee is over a \emph{set} of initial conditions and over \emph{all} admissible disturbance sequences, so it is a statement about the mission itself and not about any one simulated trajectory. Second, the disturbance appears in exactly two places, a Minkowski sum in \eqref{eq:tube} and a support-function tightening in \eqref{eq:Oinf}, and nowhere else. Nothing in Lemma~\ref{thm:cgr} changes, and, as the next section shows, neither does the combinatorial size of the reachability problem.

% =====================================================================
\section{Exact Closed-Loop Reachability with Hybrid Zonotopes}
	\label{sec:hz}

	Deciding \eqref{eq:certcond} calls for the reachable sets of \eqref{eq:tube} themselves, and the representation in which they are computed decides whether the answer is a provable certificate. Hybrid zonotopes are adopted since they represent the closed loop of a multi-parametric program without approximation.

	\subsection{Hybrid zonotopes}
		\label{sec:hzbg} A hybrid zonotope \cite{bird2023hz} is the set
		\begin{equation}
			\label{eq:hz}
			\begin{aligned}
				\mathcal{Z} = \big\{ \mathbf{G}_c \boldsymbol{\xi}_c &+ \mathbf{G}_b \boldsymbol{\xi}_b + \mathbf{c} \ : \ \boldsymbol{\xi}_c \in [-1,1]^{n_{gc}}, \\
				&\boldsymbol{\xi}_b \in \{-1,1\}^{n_{gb}},\ \mathbf{A}_c \boldsymbol{\xi}_c + \mathbf{A}_b \boldsymbol{\xi}_b = \mathbf{b} \big\} ,
			\end{aligned}
		\end{equation}
		a union of $2^{n_{gb}}$ constrained zonotopes \cite{scott2016czono} represented without enumerating them. The class is closed under linear maps, Minkowski sums and generalized intersections, and set queries over it (support, emptiness, containment) are mixed-integer programs whose branch-and-bound cost is governed by the number of binary generators $n_{gb}$. For a zonotope $\Wset$ with generators $\mathbf{G}_\Wset \in \R^{n_x \times n_w}$ and center $\mathbf{c}_\Wset$, the sum $\mathcal{Z} \oplusm \Wset$ has data
		\begin{equation}
			\label{eq:msum}
			\left( \begin{bmatrix} \mathbf{G}_c & \mathbf{G}_\Wset \end{bmatrix},\ \mathbf{G}_b,\ \mathbf{c} + \mathbf{c}_\Wset,\
			\begin{bmatrix} \mathbf{A}_c & 0 \end{bmatrix},\ \mathbf{A}_b,\ \mathbf{b} \right) ,
		\end{equation}
		so $n_{gb}$, $\mathbf{G}_b$ and $\mathbf{A}_b$ are untouched and the operation adds $n_w$ continuous generators and no binary ones.

		The multi-parametric program itself admits a hybrid-zonotope representation. The graph of the explicit solution \eqref{eq:pwa} over the parameter box, that is, the set of parameter-optimizer pairs
		\begin{equation}
			\label{eq:zstar}
			Z^\star := \left\{ (\mathbf{x}, \boldsymbol{\eta}^\star(\mathbf{x})) \ : \ \mathbf{x} \in \Theta \right\} \subset \R^{n_x + n_\eta} ,
		\end{equation}
		is a hybrid zonotope obtained from the KKT conditions of \eqref{eq:condensed} with big-$M$ complementarity, so the closed-loop reachable set of the resulting MPC law can be propagated without approximation \cite{bird2022hz}. The construction is implemented in zonoLAB \cite{koeln2023zonolab}, and related work extends it to nonlinear and neural-network controllers \cite{siefert2023hz} and tightens the relaxations of the representation itself \cite{glunt2025sharp}.

	\subsection{Hybrid-zonotope propagation}
		With $Z^\star$ from \eqref{eq:zstar}, the nominal closed-loop map is
		\begin{equation}
			\label{eq:prop}
			\Phi(\Rset_k) = \mathbf{M}_{\mathrm{next}} \left( Z^\star \cap_{[\mathbf{R}_{\mathrm{sel}}]} \Rset_k \right) ,
		\end{equation}
		a generalized intersection followed by a linear map, in which $\mathbf{R}_{\mathrm{sel}} = [\,\mathbf{I}_{n_x} \ \ \mathbf{0}\,]$ takes the state block of $Z^\star$ so it can be matched against $\Rset_k$, and $\mathbf{M}_{\mathrm{next}} = [\,\mathbf{A} \ \ \mathbf{B}\mathbf{E}_0\mathbf{T}\,]$ carries a surviving pair $(\mathbf{x}, \boldsymbol{\eta}^\star)$ to its successor $\mathbf{A}\mathbf{x} + \mathbf{B}\mathbf{E}_0\mathbf{T}\boldsymbol{\eta}^\star$. Each step adds $n_c$ binary generators, one per inequality constraint of \eqref{eq:condensed}, so $n_{gb} = k\,n_c$ grows linearly in $k$.

		Additive bounded disturbances are accounted for by taking the Minkowski sum of the propagated set with a hybrid zonotope $\Wset$ at each step, $\Rset_{k+1}^\Wset = \Phi(\Rset_k^\Wset) \oplusm \Wset$ of \eqref{eq:tube} \cite{bird2022hz}, which is \eqref{eq:msum} applied once per step.

		Notice that the Minkowski sum with the disturbances results in a hybrid zonotope with the same number of binary variables as the nominal one, so the containment check runs over the same binary dimension and its complexity is unchanged. The disturbance enters the plant of \eqref{eq:plant} and not the prediction model of \eqref{eq:ocp}, so the controller of Section~\ref{sec:empc} stays the nominal one and is verified after the fact instead of being redesigned, which is the position argued in \cite{bird2022hz}. Tube MPC \cite{specht2023tube}, the standard alternative for handling disturbances, instead tightens the constraints of \eqref{eq:ocp} by a margin sized for the worst case and applies that tightened controller at every step, whether or not the disturbance realizes. Applied here it would change the partition, the gain $\mathbf{K}_T$, and the region $\Cset_T$ themselves, which are the objects treated in Sections~\ref{sec:cgr} and~\ref{sec:cert}, instead of leaving them fixed and treating the disturbance separately.

	\subsection{Deciding containment}
		It remains to decide \eqref{eq:certcond}. Writing $\Oinf^\Wset = \{\mathbf{x} : \mathbf{a}_i^\top \mathbf{x} \leq b_i,\ i = 1,\dots,n_o\}$, containment is equivalent to $n_o$ emptiness tests,
		\begin{equation}
			\label{eq:empty}
			\begin{gathered}
				\Rset_k^\Wset \subseteq \Oinf^\Wset \quad \Longleftrightarrow \\[3pt]
				\Rset_k^\Wset \cap \left\{ \mathbf{x} : \mathbf{a}_i^\top \mathbf{x} \geq b_i + \epsilon \right\} = \emptyset , \quad i = 1,\dots,n_o ,
			\end{gathered}
		\end{equation}
		with $\epsilon > 0$ a numerical margin on the strict inequality, each a mixed-integer feasibility problem over \eqref{eq:hz}. Tightening $\mathbf{h}$ in \eqref{eq:Oinf} by $\epsilon$ before the test absorbs that margin, so an \textsc{infeasible} return proves containment in $\Oinf^\Wset$ itself. Nothing was relaxed anywhere in the propagation, so the answer is a proof and not an estimate. A feasible return is also a proof, of the opposite statement, exhibiting a witness state that reaches $\Rset_k^\Wset$ outside $\Oinf^\Wset$.

		Refuting containment typically settles quickly, since a single witness suffices, so the steps before $k^\star$ clear almost immediately and the search concentrates its cost on the step at which containment first holds. Because $n_{gb} = k\,n_c$ grows linearly with the step index, that cost is governed by the product of the horizon's constraint count and the reduction step, so a horizon short enough to keep $n_c$ small keeps the search within a practical budget.

\section{Simulation Results}
	\label{sec:results}

	The two problems of Section~\ref{sec:twoproblems} are answered in turn, the first by measuring what each parametrization costs to build, to store and to evaluate, and the second by computing after how many steps the resulting law reduces to its constant gain over the whole dispersion. The parametrization selected against Problem~\ref{prob:impl} is then carried into the reachability computation of Section~\ref{sec:hz}.

	\subsection{Scenario}
		A chaser starts \SIrange{1400}{2000}{\meter} behind a target in a circular low Earth orbit and closes to a docking box of \SI{\pm 20}{\meter}, and the full configuration is in Table~\ref{tab:scenario}. The sampling period is set from the Hohmann transfer time $T_H = \pi\sqrt{a^3/\mu} = \SI{3022}{\second}$, split into $N_{T_s}$ intervals, $T_s = \lceil T_H/N_{T_s} \rceil$, and the horizon is then $N = N_{T_s}+1$, one sample longer so the controller's prediction spans the whole transfer and converges at the far end. Since $N_{T_s}$ is used only to derive $T_s$, the rest of the paper refers to the horizon by $N$ alone.

		Position and velocity entries in Table~\ref{tab:scenario} are ordered (along-track, cross-track, nadir) in the LVLH frame, per \eqref{eq:lvlh}, and the weights listed there are common to every controller, with $\mathbf{R}$ scaled by a single scalar $\alpha$ per parametrization until the mission $\Delta v$ matches, so that all comparisons are made at equal propellant. All quantities are computed over the full state. The parametrization study additionally runs the nonlinear plant in the loop, so that constraint compliance is measured against true motion and not against the prediction model. Cross-track motion decouples in the CW model and is confined to a \SI{\pm1}{\meter} band, so it plays no part in what follows.

		\begin{table}[!t]
			\caption{\label{tab:scenario} Scenario configuration.}
			\centering
			\footnotesize
			\setlength{\tabcolsep}{3pt}
			\begin{tabular}{@{}lll@{}}
				\hline
				quantity & symbol & value \\\hline
				semi-major axis     & $a$       & \SI{7171}{\kilo\meter} \\
				grav. parameter     & $\mu$     & \SI{3.986e14}{\meter\cubed\per\second\squared} \\
				mean motion         & $n$       & \SI{1.0398e-3}{\radian\per\second} \\
				transfer time       & $T_H$     & \SI{3022}{\second} \\
				sampling period     & $T_s$     & $\lceil T_H/N_{T_s}\rceil$: 605, 303, 202, 61 s \\
				horizon             & $N$       & $N_{T_s}+1$: 6, 11, 16, 51 \\
				mission duration    &           & one transfer, $N$ steps of $T_s$ \\[3pt]
				position box        & $\Xset$   & $x\in[-100,2100]$, $z\in[-100,600]$ m \\
				velocity box        &           & $v_x, v_z \in \pm5$ m/s \\
				cross-track box     &           & $y\in\pm1$ m, $v_y\in\pm0.1$ m/s \\
				terminal box        & $\Xset_N$ & $\pm\SI{20}{\meter}$, $\pm\SI{1}{\meter\per\second}$ \\
				impulse bound       & $\overline{\mathbf{u}}$ & \SI{3}{\meter\per\second} per axis and step \\[3pt]
				dispersion          & $\Xset_0$ & $x\in[0,600]$, $z\in[-90,200]$ m \\
				                    &           & $\pm\SI{2}{\meter\per\second}$ \\
				target state        & $\mathbf{x}_{\mathrm{tgt}}$ & $[2000,0,0,0,0,0]^\top$ m, m/s \\[3pt]
				state weight        & $\mathbf{Q}$   & $\mathrm{diag}(6\mathbf{I}_3,\, 6\!\times\!10^4 \mathbf{I}_3)$ \\
				input weight        & $\mathbf{R}$   & $\alpha \cdot 5.4\!\times\!10^{6}\, \mathbf{I}$ \\
				terminal weight     & $\mathbf{Q}_N$ & DARE solution, $(\mathbf{A},\mathbf{B},\mathbf{Q},\mathbf{R})$ \\
				matched propellant  &           & \SI{2.5}{\meter\per\second} summed impulse \\
				Monte Carlo samples &           & 100 \\\hline
			\end{tabular}
		\end{table}

	% ---------------------------------------------------------------------
	\subsection{Choosing the parametrization (Problem~\ref{prob:impl})}
		Sweeping the parametrizations of Section~\ref{sec:mpc} at matched propellant, so that a difference in partition size reflects the basis and not a more aggressive controller, gives Fig.~\ref{fig:complexity}. The $a_L$ columns sweep the Laguerre pole value at $n_L = 2$, the $n_L$ columns vary the number of basis functions at $a_L = 0.4$, and $MB_2$ and $MB_3$ are move-blocking with blocks of two and three.

		\begin{figure}[!t]
			\centering
			\includegraphics[width=\columnwidth]{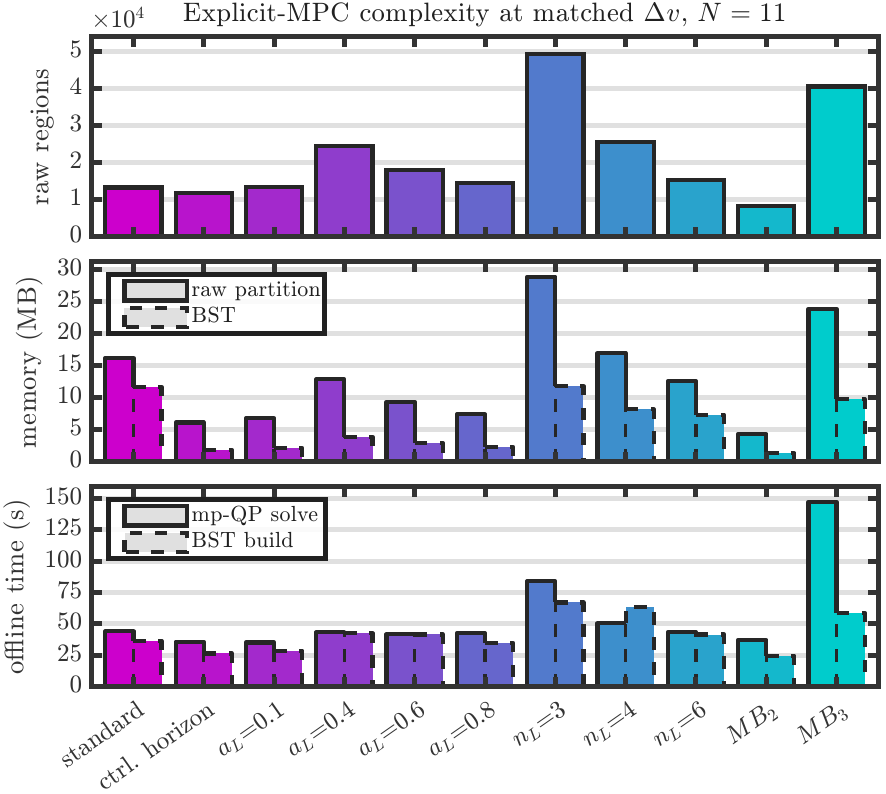}
			\caption{\label{fig:complexity} Offline partition, $\boldsymbol{\eta}$ memory and offline solve time across parametrizations at matched $\Delta v$, $N=11$.}
		\end{figure}

		The columns of Fig.~\ref{fig:complexity} are not all independent designs, since at $n_L = N$ the Laguerre matrix $\mathbf{T}$ becomes square and invertible, so $\mathbf{U} = \mathbf{T}\boldsymbol{\eta}$ spans the whole input space and the parametrized program is a change of variables on the standard one. This was checked at $N=6$, where $n_L = N$ is reached, and region count, memory and closed-loop cost come out equal to the standard column to the digit, at $a_L = 0.4$ and at every other pole tested.

		Across those columns the region count is a poor proxy for the memory a stored $\boldsymbol{\eta}$ occupies, since at $N=11$ the control horizon has \SI{11}{\percent} fewer regions than the standard design but a lookup tree \num{6.2} times smaller, and at $N=6$ the two produce almost identical partitions, \num{2356} against \num{2359} regions, for trees differing by a factor of three. That cost is driven by the dimension of the affine law held at each leaf, not by the number of leaves.

		$MB_2$ gives the smallest partition and the smallest tree of any basis tested (\num{8426} regions, \SI{1.4}{\mega\byte}), but cannot fly the mission, reaching only \SI{1.68}{\meter\per\second} of the \SI{2.5}{\meter\per\second} target at any weighting and completing 10 of 100 runs. Relaxing the blocking to $MB_3$ recovers the transfer on 65 of 100 runs, still short of every other basis tested, and costs \num{40604} regions, three times the standard partition. Blocking the input coarsely enough to shrink the partition removes the authority needed to close the transfer, and relaxing it until the vehicle flies gives no partition size back. Move-blocking is therefore the weakest of the bases tested.

		The Laguerre columns lose in a different way, since a non-zero pole costs regions, moving from the control horizon to $a_L = 0.4$ doubling the partition at fixed $n_\eta$. The Laguerre basis spreads each degree of freedom over the whole horizon, which densifies $\mathbf{G}_T = \mathbf{G}\mathbf{T}$ and admits many more candidate active sets, whereas the control-horizon basis zeroes the tail and thereby renders a large block of input-bound rows redundant. The conditioning penalty that \cite{rossiter2023review} attributes to the Laguerre basis has the same origin. Compliance is unchanged across poles in $[0.1, 0.6]$, so the regions the pole adds return neither a lower closed-loop cost nor better compliance. The control horizon is therefore the preferable basis of those tested, giving the lowest closed-loop cost from $N=11$ onward and the smallest stored law of the bases that complete the transfer.

		Compiling each partition into the binary search tree of Section~\ref{sec:empc}, and comparing the result against the standard controller, gives Table~\ref{tab:complexity}. Every row is computed on the six-state model against a nonlinear plant over \num{100} Monte Carlo initial conditions, the offline column is the multi-parametric solve plus the tree build, and $J$ is the value of the cost function accumulated over the trajectory, weighted by the sampling period. The number of free variables is $n_\eta = 3N$ for the standard basis and $n_\eta = 6$ for the control horizon.

		\begin{table}[!t]
			\caption{\label{tab:complexity} Standard versus control-horizon parametrization at matched $\Delta v$.}
			\centering
			\footnotesize
			\setlength{\tabcolsep}{3.5pt}
			\begin{tabular}{@{}llrrrrr@{}}
				\hline
				$N$ & basis & regions & $\boldsymbol{\eta}$ (MB) & $\mathbf{u}_0$ (MB) & offline (s) & $J$ ($\times 10^{10}$) \\\hline
				6  & standard        & \num{2359}   & \num{1.14} & \num{0.19} & 51.9   & \num{2.911} \\
				& control horizon & \num{2356}   & \num{0.38} & \num{0.19} & 46.2   & \num{2.914} \\[2pt]
				11 & standard        & \num{13332}  & \num{11.8} & \num{1.08} & 81.0   & \num{2.360} \\
				& control horizon & \num{11854}  & \num{1.90} & \num{0.95} & 62.5   & \num{2.197} \\[2pt]
				16 & standard        & \num{25774}  & \num{33.0} & \num{2.07} & 131.5  & \num{2.246} \\
				& control horizon & \num{22972}  & \num{3.69} & \num{1.85} & 88.8   & \num{1.906} \\[2pt]
				51 & standard        & \num{122498} & \num{500}  & \num{9.85} & 2199.7 & \num{2.129} \\
				& control horizon & \num{131131} & \num{21.0} & \num{10.5} & 707.4  & \num{1.660} \\\hline
			\end{tabular}
		\end{table}

		The two designs separate on how much of the predicted sequence a leaf holds. Storing all of it costs $n_\eta$ numbers per leaf, which grows as $N n_u$ for the standard basis while the control horizon stays at $N_u n_u = 6$, so that column grows with $N/N_u$, reaching \num{24}$\times$ at $N=51$. Only $\mathbf{u}_0$ is applied, however, and a leaf holding it alone stores one affine law, $\mathbf{u}_0 = \boldsymbol{\Gamma}_i^{(1)} \mathbf{x} + \boldsymbol{\gamma}_i^{(1)}$ with $\boldsymbol{\Gamma}_i^{(1)} \in \R^{n_u \times n_x}$, which is \num{21} entries whatever the parametrization and brings the two within \SI{15}{\percent} of each other at every horizon. The control horizon keeps its margin elsewhere, building in \SI{707}{\second} against \SI{2200}{\second} at $N=51$, and its $J$ comes out up to \SI{22}{\percent} lower at equal propellant.

		\subsubsection{Does an explicit law fit on board?}
			\label{sec:memory}

			Memory is the objection raised against explicit MPC in space, the partition being expected to outgrow flight storage at a realistic horizon. Every size quoted here is for a law stored in single precision. At $N=16$ the control-horizon controller is \SI{1.85}{\mega\byte} and takes \SI{88.8}{\second} to build\footnote{Offline solve on an AMD Ryzen~7 5700X3D (8 cores/16 threads, \SI{3.0}{\giga\hertz} base / \SI{4.1}{\giga\hertz} boost, 32~GB RAM).}. The flight computer of the Aalto-1 CubeSat carries flash memories of \SI{256}{\mega\byte} and \SI{8}{\mega\byte} alongside \SI{32}{\mega\byte} of SDRAM \cite{araujo2025cots}, so the law fits the smaller of the two flash memories. At $N=51$ the transfer is resolved at $T_s = \SI{61}{\second}$, far finer than the maneuver requires, and the law is \SI{10.5}{\mega\byte}, which the SDRAM carries at run time.

			The tree at $N=51$ is at most \num{21} levels deep, each level costing one inner product in $\R^{6}$, and the leaf adds a single $3 \times 6$ matrix-vector product, for \num{291} floating-point operations in all. Compiled to C and measured on the same workstation, one evaluation takes \SI{86}{\nano\second} on average and under \SI{300}{\nano\second} at worst. The Aalto-1 computer runs an AT91RM9200 rated at \num{200}~MIPS, which places the same evaluation in the microseconds with hardware floating point and in the tens of microseconds under software emulation. Both figures are at least six orders below the sampling period. The on-board cost of an explicit law is therefore set by what each leaf holds and by how deep the tree is, and for this problem both are within a CubeSat flight computer.

			The control-horizon parametrization with $N_u = 2$ therefore answers Problem~\ref{prob:impl}, and is the design carried into Problem~\ref{prob:verif}, whose constant-gain region and reachable tube are computed next.

	% ---------------------------------------------------------------------
	\subsection{Verifying the resulting controller (Problem~\ref{prob:verif})}

		\subsubsection{The constant-gain region and its robust invariant core}
			Every result in this section is for the control-horizon design with $N_u = 2$ carried forward from Problem~\ref{prob:impl}, at $N = 11$ and $T_s = \SI{303}{\second}$. Lemma~\ref{thm:cgr} gives $\Cset_T$ in closed form with 20 facets, and both it and $\Oinf$ are obtained in under a second, without touching the partition built for Problem~\ref{prob:impl}. The Gilbert-Tan recursion terminates at the third iteration on an $\Oinf$ of 18 facets, \SI{67}{\percent} the volume of $\Cset_T$ (volumes taken on the vertex representation of each polytope), so the constant-gain region is not itself invariant here and $\Cset_T$ alone will not carry the argument, which is why Theorem~\ref{thm:rpi} is needed. The same recursion is exact under disturbance, with the tightening in \eqref{eq:Oinf} strict whenever $\Wset \neq \{0\}$.

		\subsubsection{The reachable tube}
			Propagating the full dispersion forward under the closed loop gives the reachable sets shown in Fig.~\ref{fig:tube}, drawn in the position and velocity planes against $\Cset_T$, with the tube colored by step from $\Rset_0$ in magenta through to $\Rset_{10}$ in cyan.

			\begin{figure}[!t]
				\centering
				\includegraphics[width=\columnwidth]{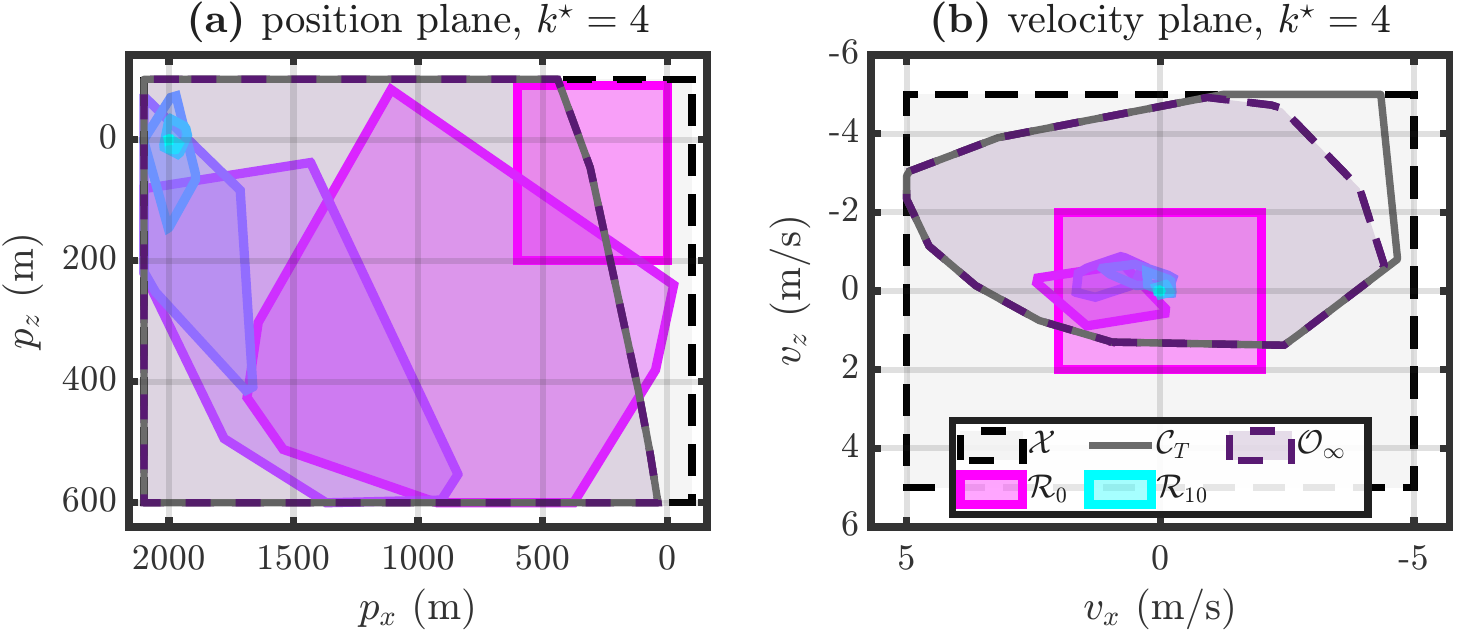}
				\caption{\label{fig:tube} Reachable sets $\Rset_k$ of the closed loop against the constant-gain region $\Cset_T$. (a) position plane. (b) velocity plane.}
			\end{figure}

			The projections drawn above are shadows, and containment of a shadow is necessary but not sufficient, so every step index is taken from the full-dimensional computation. The two projections cross into $\Cset_T$ at different steps. The full-dimensional test establishes the hypothesis of Theorem~\ref{thm:cert} at $k^\star = 4$ over the whole dispersion, the program remaining feasible on the tube at every earlier step as Assumption~\ref{as:feas} requires.

			Beyond $k^\star$, Theorem~\ref{thm:cert} gives the whole statement at once, with the applied input the fixed gain $-\mathbf{K}_T \mathbf{x}$, the state confined to $\Oinf \subseteq \Cset_T \subseteq \Xset$ so that constraints hold for all time, and stability decided by the spectral radius $\rho(\Acl) < 1$, which holds here with $\rho(\Acl) = \num{0.54}$, so the state converges to the minimal robust invariant set of $(\Acl, \Wset)$ and to the origin when $\Wset = \{\mathbf{0}\}$. The on-board partition is therefore exercised over only four steps, and past that the applied law is a single matrix multiplication.

		\subsubsection{Robustness}
			Under Assumption~\ref{as:W} the same conclusion holds for every disturbance sequence in $\Wset$, with $\Oinf$ replaced by the tightened set $\Oinf^\Wset$ of Theorem~\ref{thm:rpi} and the tube by \eqref{eq:tube}. By \eqref{eq:msum} the disturbance contributes only continuous generators, so the disturbed problem is solved over the same binary dimension as the nominal one, the entire cost of robustness being a right-hand-side tightening in \eqref{eq:Oinf} and $n_w$ extra columns in $\mathbf{G}_c$. If $\Wset$ is large enough that $\Oinf^\Wset$ is empty, the recursion returns the empty set and Theorem~\ref{thm:cert} does not apply.

			Since $\mathbf{w}_k$ in \eqref{eq:plant} is the error committed between one control instant and the next, $\Wset$ has to bound the residual accumulated over a single sampling period and not over the mission. The measurement adopted here mirrors that definition, Theorem~\ref{thm:cert} taking $\Wset$ as given so that any construction satisfying Assumption~\ref{as:W} may be used. From each state of a grid over $\Xset$, the nonlinear propagator and the prediction model of \eqref{eq:cwmat} are both advanced across one $T_s$, and the difference between the two end states is the disturbance that would have acted on that step.

			Sweeping that measurement over the separation $\ell$ gives Fig.~\ref{fig:cwbudget} and shows how the residual grows with distance. Every point places the chaser where a single Hohmann transfer to the target would begin. Errors are scored as $\lVert [\, \Delta \mathbf{p}^\top \ \ \Delta \mathbf{v}^\top / n \,] \rVert$, the weighting used for model comparison in \cite{sullivan2017survey}, so position and velocity carry the same units, and the truth model is turned on one term at a time to separate their contributions. Differential $J_2$ dominates and grows linearly in $\ell$, while the linearization error grows quadratically, so the two meet at \SI{20}{\kilo\meter} and the linearization error dominates past it, well outside the close-range regime of Remark~\ref{rem:scope}. Cutting the transfer into more intervals shrinks all of it, from \SI{3.83}{\meter} at $T_s = T_H/5$ to \SI{0.31}{\meter} at $T_H/50$, and the control-horizon basis selected against Problem~\ref{prob:impl} keeps the resulting horizon affordable at the \SI{10.5}{\mega\byte} of Table~\ref{tab:complexity}.

			\begin{figure}[!t]
				\centering
				\includegraphics[width=\columnwidth]{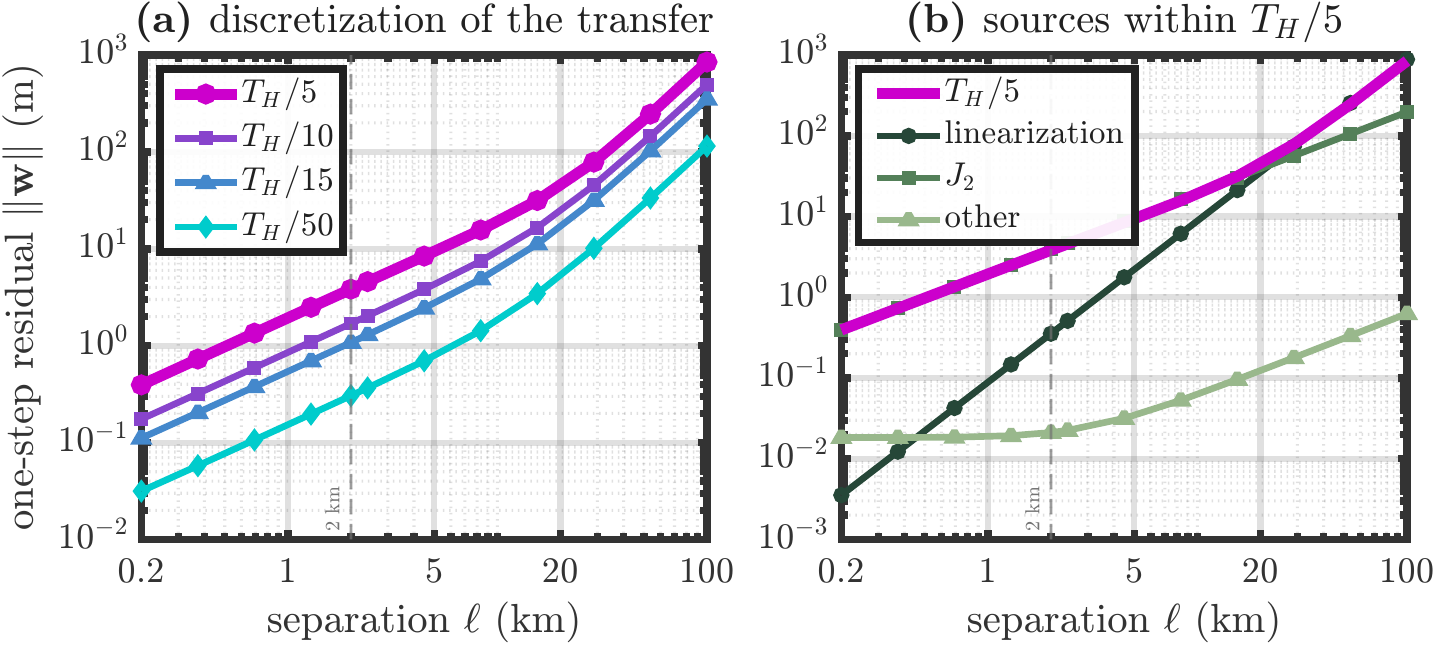}
				\caption{\label{fig:cwbudget} One-step residual against separation. (a) transfer split into $N_{T_s}$ intervals. (b) sources within $T_H/5$.}
			\end{figure}

			Reading Fig.~\ref{fig:cwbudget} at the \SI{2}{\kilo\meter} start of the transfer, marked on both panels, differential $J_2$ and the linearization error contribute \SI{3.91}{\meter} and \SI{0.35}{\meter}, while the higher zonals, differential drag, solar radiation pressure and third-body attraction together add \SI{0.02}{\meter}. Taking the componentwise maximum over the grid with a \num{1.25} margin gives
			\begin{equation}
				\label{eq:Wbox}
				\Wset = \{\mathbf{w} : \lvert \mathbf{w} \rvert \leq \bar{\mathbf{w}}\} , \quad
				\bar{\mathbf{w}} = \begin{bmatrix}
					\SI{2.09}{\meter} \\ \SI{0.0135}{\meter} \\ \SI{3.14}{\meter} \\
					\SI{0.0127}{\meter\per\second} \\ \SI{3.3e-5}{\meter\per\second} \\ \SI{0.0157}{\meter\per\second}
				\end{bmatrix} .
			\end{equation}
			The cross-track entries sit two orders of magnitude below their in-plane counterparts, so the in-plane/out-of-plane decoupling of \eqref{eq:cwmat} carries over to the disturbance. Repeating the same measurement at geostationary altitude, where the dominant perturbations differ \cite{spiridonova2016geo}, gives a residual \num{227} times smaller, so low Earth orbit is the demanding case and $\bar{\mathbf{w}}$ is conservative for the higher one.

			The results reported here take $5\bar{\mathbf{w}}$ instead of $\bar{\mathbf{w}}$, a per-step position error of \SIrange{10}{16}{\meter} against a \SI{20}{\meter} arrival box. Containment still holds, at $k^\star = 6$ against the nominal \num{4}, with $\Oinf^\Wset$ retaining \SI{33.6}{\percent} of $\Cset_T$, so five times the physical disturbance costs two steps of $k^\star$, and Fig.~\ref{fig:robustslice} draws the disturbed tube against both sets in the views of Fig.~\ref{fig:tube}.

			The disturbed sets are larger than their nominal counterparts because $\Wset$ enters \eqref{eq:tube} at every step, and where the nominal tube already lies against a state constraint the disturbed one crosses it by one $\bar{\mathbf{w}}$, which at $5\bar{\mathbf{w}}$ is \SI{15.7}{\meter} in $z$ and \SI{10.5}{\meter} in $x$. The crossing is a single step's displacement and not an accumulation, since the displaced state is the one the controller acts on at the next instant. It also stops at $k^\star$, because from there on the state lies in $\Oinf^\Wset$, whose rows already carry the tightening of \eqref{eq:Oinf}. Both figures are drawn against the box of Table~\ref{tab:scenario}, so what the disturbed tube crosses is the margin Assumption~\ref{as:margin} reserves and not the corridor itself. The box is left without removing the disturbance effect here so that Figs.~\ref{fig:tube} and~\ref{fig:robustslice} are read against the same set.

			\begin{figure}[!t]
				\centering
				\includegraphics[width=\columnwidth]{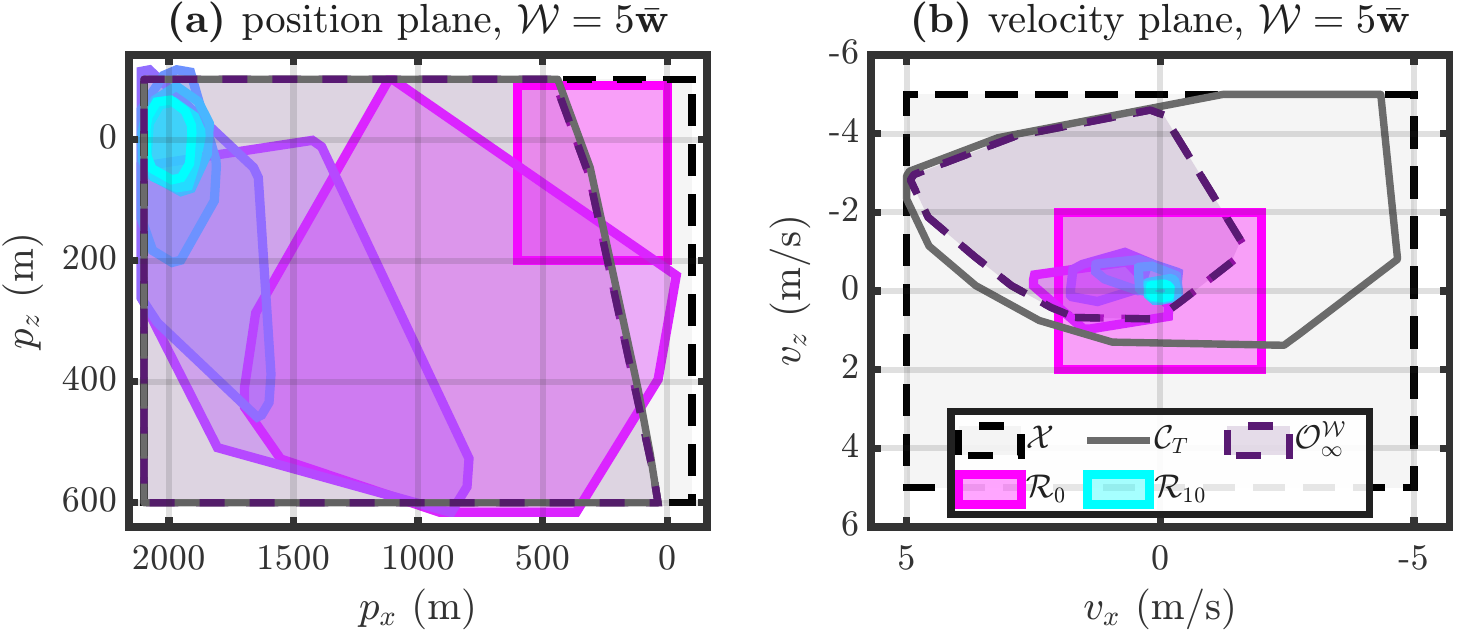}
				\caption{\label{fig:robustslice} Disturbed reachable sets $\Rset_k^\Wset$ at $5\bar{\mathbf{w}}$ against $\Cset_T$ and the tightened $\Oinf^\Wset$. (a) position plane. (b) velocity plane.}
			\end{figure}

		\subsubsection{Effect of the parametrization on the reduction step}

			Because $\Cset_T$ is available in closed form for every basis, Theorem~\ref{thm:cert} can be reapplied to each at no cost. Doing so at matched $\Delta v$ gives the volume of $\Oinf$ in Fig.~\ref{fig:sweep}, with the reduction step $k^\star$ annotated for each parametrization. Lemma~\ref{thm:cgr} and Theorem~\ref{thm:cert} serve every basis without modification, move-blocking included, although it is already ruled out on mission completion.

			The Laguerre basis at $n_L=2$, $a_L=0.3$ reduces earliest of all, at $k^\star = 3$ against the standard basis's 4, for a cost penalty of \SI{3.1}{\percent}. Every basis except move-blocking reduces within one step of the standard design, the control horizon carried forward from Problem~\ref{prob:impl} among them at a penalty of \SI{1.7}{\percent}. Move-blocking is the outlier, entering four steps later than the rest at a \SI{25}{\percent} cost penalty. It also holds the smallest region of the sweep while entering last, against the standard basis which holds the largest and enters at $k^\star = 4$, so $k^\star$ is set not by the extent of $\Oinf$ but by how quickly the induced closed loop $\Acl$ drives the reachable tube toward it.

			\begin{figure}[!t]
				\centering
				\includegraphics[width=0.6\columnwidth]{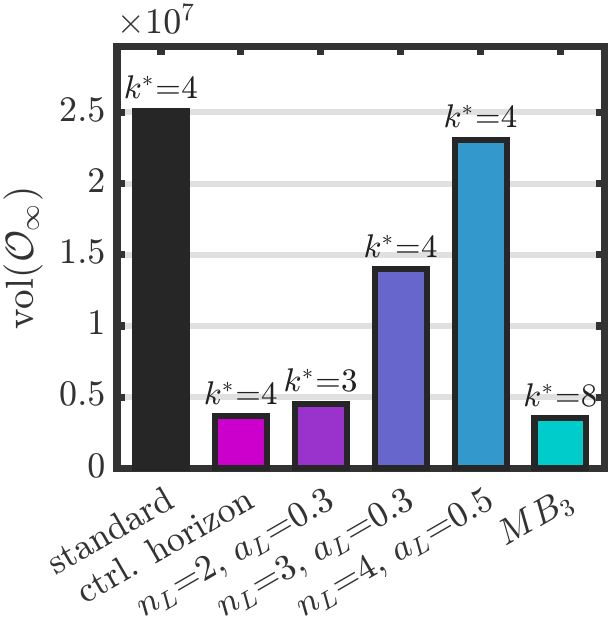}
				\caption{\label{fig:sweep} Volume of the robust invariant set $\Oinf$ per parametrization.}
			\end{figure}

% =====================================================================
\section{Conclusion}

	Explicit MPC was treated here as a design problem and as a verification problem, both of which are governed by the input parametrization. Among the parametrizations available for reducing explicit complexity, a short control horizon lowered the closed-loop cost by up to \SI{22}{\percent} at matched propellant and built three times faster. Its advantage in stored size depends on what a leaf keeps, reaching \num{24}$\times$ when the whole predicted sequence is held and closing to within \SI{15}{\percent} when storing only the applied action, since the memory a partition then occupies is fixed by the dimension of that one law and by how many distinct laws survive merging. Move-blocking was the weakest of the bases tested, since blocking the input coarsely enough to shrink the partition removed the authority needed to close the transfer, and relaxing it until the vehicle flew gave no partition size back. At $N=51$ the law a flight computer must carry is \SI{10.5}{\mega\byte} and evaluates in under \SI{300}{\nano\second}.

	The controller obtained is then certifiable offline, being proven to reduce to a fixed gain. For any linear input parametrization, the set on which the applied input equals a single constant gain is available in closed form from the condensed program, without a multi-parametric solve or an enumeration of critical regions. Its maximal robustly invariant subset then follows from the Gilbert-Tan recursion with a support-function tightening. Exactly propagating the closed-loop reachable set on the hybrid-zonotope representation of the multi-parametric program and testing containment in that subset returned $k^\star = 4$ for the low-Earth-orbit scenario treated here, beyond which the mission is a linear system whose constraint satisfaction and stability follow from a set containment and from $\rho(\Acl) = \num{0.54}$. The statement holds for every initial condition in the dispersion and every disturbance sequence in $\Wset$, and the disturbance enters through two Minkowski sums that leave the binary dimension of the test unchanged, so five times the measured disturbance cost two steps of $k^\star$ at no increase in the cost of deciding it.

	The result carries two limits. The invariant set is maximal within the constant-gain region and not maximal for the closed loop, so $k^\star$ is an upper bound and a larger invariant set would only lower it. The prediction model is the CW one, which confines the scope to a near-circular target at close range, everything outside that scope having to enter as a larger $\Wset$ rather than as a different plant. Neither limit changes the form of the result, so Lemma~\ref{thm:cgr} and Theorem~\ref{thm:cert} apply unchanged to a more demanding scenario. Together the two results supply a rendezvous controller whose storage, evaluation time and closed-loop behavior are all fixed before launch.

\bibliographystyle{IEEEtran}
\bibliography{sample}

\end{document}